\documentclass{article}[12pt]
\usepackage{amsmath}

\usepackage{graphicx}
\usepackage{amsthm}
\usepackage{amsfonts}

\usepackage{color}

\newcommand{\ds}{\displaystyle}
\newcommand{\comment}[1]{}

\newcommand{\mx}[1]{\mathbf{#1}}

\usepackage{bbm}
\newcommand{\1}{\mathbbm{1}}
\newcommand{\ones}{\1}

\newcommand{\ii}{\mathbbm{i}}
\newcommand{\av}{\boldsymbol{\alpha}}
\renewcommand{\v}{\boldsymbol{v}}
\renewcommand{\u}{\boldsymbol{u}}
\newcommand{\w}{\boldsymbol{w}}
\newcommand{\bv}{\boldsymbol{\beta}}
\newcommand{\ev}{\boldsymbol{e}}
\newcommand{\gv}{\boldsymbol{\gamma}}

\newcommand{\A}{\mx{A}}
\newcommand{\B}{\mx{B}}
\newcommand{\G}{\mx{G}}
\renewcommand{\H}{\mx{H}}
\newcommand{\I}{\mx{I}}
\newcommand{\J}{\mx{J}}
\newcommand{\M}{\mx{M}}
\renewcommand{\P}{\mx{P}}
\newcommand{\Q}{\mx{Q}}
\newcommand{\R}{\mx{R}}
\newcommand{\W}{\mx{W}}

\newtheorem{theorem}{{\bf Theorem}}
\newtheorem{example}{{\bf Example}}

\newtheorem{definition}{{\bf Definition}}
\newtheorem{lemma}[theorem]{Lemma}

\newcommand{\Figref}[1]{Figure \ref{#1}}

\newlength{\LI}
\title{A constructive proof of the phase-type characterization theorem}

\author{
Ill\'es Horv\'ath$^1$ and Mikl\'os Telek$^{2,3}$  \\[3mm]
$^1$ {\small  MTA-BME Information Systems Research Group, Budapest, Hungary} \\
$^2$ {\small Department of Telecommunications, Technical University of Budapest, Budapest, Hungary}\\
$^3$ {\small Inter-University Center of Telecommunications and Informatics, Debrecen, Hungary}\\
{\small e-mail: horvath.illes.antal@gmail.com,  telek@hit.bme.hu } \\
}

\begin{document}
\sloppy

\maketitle

\begin{abstract}

The paper presents a new proof of O'Cinneide's characterization theorem \cite{OCIN90}. It is much simpler than the original one and
constructive in the sense that we not only show the existence of a phase type representation, but present a procedure which creates
a phase type representation. We prove that the procedure succeeds when the conditions of the characterization theorem hold.

\end{abstract}

\section{Introduction}

The characterization theorem of O'Cinneide \cite{OCIN90} proves that any finite order matrix exponential function which is strictly positive
in $(0,\infty)$ and satisfies the dominant eigenvalue condition has a finite dimensional phase type (PH) representation.
Based on this theorem Mocanu and Commault \cite{CoMo03} proposed a procedure for computing
the PH representation of such matrix exponential function.
A quite different approach from Maier \cite{Maier91} proposes a similar procedure
based on Soittola's automata-theoretic algorithms \cite{Soittola}.
All of these papers prove the characterization theorem, but
use complex mathematical concepts, such as polytopes, or positive rational sequences.
Additionally, both procedures in \cite{CoMo03} and in \cite{Maier91} are implicit in the sense that an essential parameter
($\tau$ in \cite{CoMo03} and $c$ in \cite{Maier91}) are found as a result of a numerical search.

In this paper we present a constructive proof of the characterization theorem by proposing an explicit procedure for computing a
phase type (PH) representation of a matrix exponential function and showing that the procedure always terminates successfully
if the matrix exponential function satisfies the positivity and the eigenvalue conditions.

Compared to the existing resuls, one of the main advantages of the presented constructive proof is that it
is rather elementary, using basic function and matrix theory
and stochastic interpretation of Markov processes.
Another contribution of the paper is that it links the sparse monocyclic representation \cite {CoMo03} to the characterization theorem \cite{OCIN90}.

\section{Preliminaries}
\label{sec:prelim}

\begin{definition}
\label{def:me}
Let $X$ be a non-negative random variable with
probability density function (pdf)
\[ f_X(x)=\frac{\mathrm{d}}{\mathrm{d}x}\mathrm{Pr}(X<x)=- \av \A e^{\A x} \1,  ~~~ x\geq 0,\] where
$\av$ is an initial row vector of size $n$
with $\int_0^\infty f_X(x) dx=1$ (there is no probability mass at zero), 
$\A$ is a square matrix of size $n\times n$ and $\1$ is the column vector of ones of size $n$.
In this case, we say that $X$ is matrix-exponentially distributed with representation $\av, \A$, or shortly, ME($\av, \A$)-distributed.
\end{definition}

In Definition \ref{def:me} the elements of $\av$ and $\A$ are real numbers without any specific restriction on their sign and the only
restriction is that $f_X(x)$ is non-negative for $x\geq 0$. We note that $\A$ and $e^{\A x}$ commute.

For a given representation $(\av,\A)$, the size $n$ of $\av$ (and $\A$) is called the \emph{order} of the representation.

The representation of a given ME distribution is not unique.

\begin{theorem}
\label{thm:trafo}
\cite{[BUCH10a]}
Let ME($\av, \A$) of order $n$ and ME($\gv, \G$) of order $m$ be two ME distributions with
pdf $f_X(x)$ and $f_Y(x)$, respectively. If $\A$ is $n\times n$

If there exists a matrix $\W$  of cardinality $n\times m$ such that
$$\gv = \av \W,\quad  \A \W =\W \G ,\quad  \1_n=\W \1_m,$$
then
ME($\av, \A$) $\equiv$ ME($\gv, \G$) (that is, $f_X(x)=f_Y(x)$).
\end{theorem}

\begin{proof}
\[
\begin{array}{l}
f_Y(x)=- \gv \G e^{\G x} \1_m = - \av \W \G e^{\G x} \1_m =
- \av  \A e^{\A x} \W \1_m =
- \av  \A e^{\A x} \1_n = f_X(x).
   \\
\end{array}
\]
\end{proof}


Theorem \ref{thm:trafo} will be used as a representation transformation tool.
The size of column vector $\1$ is explicitly indicated in the theorem as a subscript.

\begin{definition}
\label{def:minimal}
A representation of an ME distribution has minimal order if the distribution has no representation of a smaller order.
A representation of minimal order is referred to as a \emph{minimal representation}.
\end{definition}

In a minimal representation, there are no ``extra'' or ``redundant'' eigenvalues in matrix $\A$.
More precisely a minimal representation has the following properties \cite{TeHo07a}:
\begin{itemize}
\item[P1)] All Jordan blocks of $\A$ have different eigenvalues.
\item[P2)] All eigenvalues contribute to the distribution with maximal multiplicity. For example, a Jordan block of size $n_i$ corresponding to eigenvalue $-\lambda_i$ results in the terms $\sum_{j=1}^{n_i} c_{\lambda_i,j}x^{j-1}e^{-\lambda_ix}$ in $f_X(x)$, where $c_{\lambda_i,n_i}\neq 0$.
\item[P3)] $\av$ is not orthogonal to any of the right-eigenvectors of $\A$.
\item[P4)] $\1$ is not orthogonal to any of the left-eigenvectors of $\A$.
\item[P5)] The Jordan block structures of all minimal representations of an ME distribution are identical.
\end{itemize}

These properties are explained further in Appendix \ref{sec:appb}. { Based on these properties, a minimal representation can be obtained directly from $f_X$. If $f$ takes the form
$$f(x)=\sum_{i=1}^m\sum_{j=1}^{n_i} c_{\lambda_i,j}x^{j-1}e^{-\lambda_ix}$$
where $\lambda_i$ are different and $c_{\lambda_i,n_i}\neq 0$, then we will consider the following representation $(\av,\A)$:
$$\A=\left(
\begin{array}{ccccc}
\J_1 & 0 & \dots & & 0\\
0 & \J_2 & 0 &\dots & 0\\
\vdots\\
0 & \dots & & 0 & \J_m
\end{array}
\right),
$$
where
$$\J_i=\left(
\begin{array}{ccccc}
\lambda_i & 1 & \dots & & 0\\
0 & \lambda_i & 1 &\dots & 0\\
\vdots\\
0 & \dots & & 0 & \lambda_i
\end{array}
\right).
$$
and $\J_i$ is of size $n_i$. $\av$ can be calculated by solving
$$-\av e^{\A x}\A\ones=f_X(x);$$
this equation can be solved because the left-hand side contains all the terms $x^{j-1}e^{-\lambda_ix}$ up to $j\leq n_i$ for $i=1,\dots,m$.
\begin{lemma}
\label{lemma:odavissza}
The representation $(\av,\A)$ is minimal for $f_X$.
\end{lemma}
The proof is essentially due to properties P1-P5 and the fact that no Jordan block of size smaller than $n_i$ can represent the term $x^{n_i-1}e^{-\lambda_ix}$. Appendix \ref{sec:appb} elaborates more on this topic.
}

If representation $(\av,\mx{A})$ is minimal then
there are some straightforward necessary conditions for vector $\av$ and matrix $\mx{A}$ to define a valid distribution:
\begin{itemize}
    \item[C1)] The eigenvalues of $\mx{A}$ have negative real part (to avoid divergence of $f_X(x)$ as $x\rightarrow\infty$).
    \item[C2)] There is a real eigenvalue of $\A$ with maximal real part (to avoid oscillations to negative values as $x\rightarrow\infty$).
    \item[C3)] $\av \1 = 1$ (normalizing condition which ensures $\int_0^\infty f_X(x) dx=1$).
    \item[C4)] If for all $i\in\{0, 1, \ldots, j-1\}$ the $i$th derivative of $f_X(x)$ is zero then the $j$th derivative of $f_X(x)$ is non-negative
    (to avoid decreasing behavior around $x=0$).
\end{itemize}
If any of these necessary conditions are violated then the tuple consisting of the vector $\av$ and matrix $\A$
does not define a valid ME distribution. Note that non-minimal representations might contain any additional eigenvalues, including for example\ positive ones.

A subclass of ME distributions is the class of phase-type distributions (PH distributions).
\begin{definition}
\label{def:ph}
If $X$ is an ME($\av, \mx{A}$) distributed random variable, where
$\av$ and $ \A$ have the following properties:
\begin{itemize}
    \item $\alpha_i\geq 0$, $\av \1 = 1$
    \item $A_{ii} < 0$, $A_{ij} \geq 0$ for $i\neq j$, $\A \1 \leq 0$
    \item $ \mx{A}$ is non-singular,
\end{itemize}
then we say that $X$ is phase-type distributed with representation $(\av, \A)$,
or shortly, PH($\av, \A$) distributed.
\end{definition}

PH distributions can be interpreted as the time of absorption in a
CTMC \cite{[NEUT81]} and consequently the conditions of Definition \ref{def:ph} are sufficient
for vector $\av$ and matrix $\A$ to define a valid distribution.
Vector $\av$ or  matrix $\A$
satisfying the conditions of Definition \ref{def:ph} are referred to as {\em Markovian}.

The following properties are essential for the characterization of
ME distributions.

\begin{definition}
\label{def:dec}
An ME($\av$, $\A$) distribution satisfies the {\em dominant
eigenvalue condition} (DEC) if for some minimal representation
ME($\gv$, $\G$), $\G$ has a single eigenvalue with maximal real
part. This eigenvalue is called the dominant eigenvalue. Its
multiplicity may be higher than 1.
\end{definition}

Definition \ref{def:dec} excludes the case when $a$ is the dominant real eigenvalue and there is a pair of
complex eigenvalues with the same real part, for example $a\pm \ii b$, where $\ii$ is the imaginary unit.

Properties P1-P5 ensure that if C1-C4 hold for one minimal
representation ME($\av$, $\A$), they hold for all equivalent minimal representations.
Additionally, if the dominant eigenvalue has multiplicity higher than 1,
then it belongs to a Jordan-block whose size is equal to the multiplicity of the dominant eigenvalue.

\begin{definition}
The ME($\av$, $\A$) distribution with density $f_X$ satisfies the
{\em positive density condition} if $f_X(x)>0$ for all $x\in(0,\infty)$.
\end{definition}

\begin{theorem}\cite{OCIN90}
\label{thm:main}
If $f_X$ is ME($\av$, $\A$) distributed, then $f_X$ has a finite dimensional PH($\bv$, $\B$) representation
iff the following two conditions hold:
\begin{itemize}
\item ME($\av$, $\A$) satisfies the dominant eigenvalue condition;
\item $f_X$ satisfies the positive density condition.
\end{itemize}
\end{theorem}

The original proof of O'Cinneide in \cite{OCIN90} is rather complex, using
Laplace-Stieltjes transform and geometric properties of the space of
PH-distributions. In this paper we present an
algorithm that gives a constructive and altogether more elementary
proof, using function and matrix theory.

\section{Procedure and proof}

Our main goal is an algorithm that provides a constructive proof for
the sufficient direction of Theorem \ref{thm:main}, that is,  given
that the dominant eigenvalue condition and the positive density
condition hold for ME($\av$, $\A$), find a PH-representation equivalent to ME($\av$, $\A$);
in other words, find a vector-matrix pair ($\bv$, $\B$) where $\bv$ and $\B$ are
Markovian and define the same distribution as ME($\av$, $\A$).

This section is devoted to the algorithmic construction,
also stating the theorems used along the way. Proofs are given in Appendix \ref{sec:appc}.

{
We also included a proof for the necessary direction of Theorem \ref{thm:main} in Appendix \ref{sec:appb}.
While the proof of the necessary direction is straightforward using the techniques in \cite{OCIN90},
we opted to include a self-contained, elementary proof that is more in line with the methods of the present paper.}

\subsection{Sketch of the algorithm}

The algorithm consists of five main steps. Steps 1 and 2 are
preparatory, and Step 5 is just correction related to Step 2.
\begin{itemize}
\item Step 1. We find an equivalent minimal representation $(\av_1, \A_1)$ for ($\av$, $\A$) if it is not minimal by eliminating any ``extra'' eigenvalues of $\A$, which does not contribute to the pdf.
{ We refer to Lemma \ref{lemma:odavissza} and \cite{[BUCH10a]} for a straightforward and computationally stable method of finding a minimal representation.}

\item
Step 2. This step applies only if density is zero at $0$, that is, $f_X(0)=0$. This step is essentially what may be called ``deconvolution'': we represent $f_X$ as the convolution of some $f_Y$ matrix exponential density function with $f_Y(0)>0$ and an appropriate Erlang-distribution Erlang($k,\mu$) (see Lemma \ref{lemma:fnulla}); if $f_Y$ has a Markovian representation, then it gives a straightforward Markovian representation for $f_X$ as well (see Lemma \ref{lemma:composition}). Thus we only need to find a Markovian representation for $f_Y$ (and the corresponding representation, which is obtained from Lemma \ref{lemma:fnulla}), where $f_Y(0)>0$. If this step is applied, Steps 3 and 4 are applied for $f_Y$ instead of $f_X$, and we switch back to $f_X$ in Step 5.

\item
Step 3. An equivalent representation ($\gv$, $\G$) is given with Markovian matrix $\G$, while $\gv$ may still have negative elements. The main tool of this step is the so-called monocyclic structure { (with Feedback-Erlang blocks)}. Typically, the size of $\G$ is
larger than that of $\A_2$ (because each pair of complex conjugate eigenvalues is represented with at least 3 phases);
that said, $\G$ is a sparse matrix with a simple block bi-diagonal structure. For this step only the dominant eigenvalue condition is necessary.

\item Step 4. $\gv$ and $\G$ are transformed further into
$\bv$  and $\B$ where $\bv$ is Markovian (and the Markovity of
$\mx{B}$ is also preserved) essentially by adding an ``Erlang-tail''
(a number of sequentially connected exponential phases with identical rates) of
proper order and rate to the monocyclic structure described by the Markovian  matrix $\G$.
The main mathematical tool of this step is the approximation of elementary functions. Essentially, this last step is the contribution of the paper.
The skeleton of this step is composed of the following elements:
\begin{itemize}
\item Find $\tau$ such that $\gv e^{\G\tau}>0$ (element-wise).
Such $\tau$ always exists if the dominant eigenvalue and the positive density conditions hold and the pair $(\gv,\G)$ results from the previous step.
{ We remark that for a general representation, even if $\G$ is Markovian, such a $\tau$ may not exist. This is further explained after Lemma \ref{lemma:csillag}.}

\item Find $\lambda'$ such that
$$\gv \left(\mx{I}+\frac{\mx{G}}{\lambda}\right)^{\tau\lambda} > 0\quad \forall\lambda\geq \lambda'$$ which is always possible since
$\left\|\gv (\mx{I}+\frac{\mx{G}}{\lambda})^{\tau\lambda} - \gv e^{\mx{G}\tau} \right\| \to 0$ as $\lambda\to \infty$.

\item Let $\epsilon=\inf_{x\in(0,\tau)} f_X(x)$.
$\epsilon>0$ because of the positive density condition and the result of Step 2.
Find $\lambda''$ such that
\[\left|- \gv e^{\mx{G}\tau} \G \1 + \gv \left(\mx{I}+\frac{\G}{\lambda}\right)^{\tau\lambda} \G\1  \right|
< \epsilon \quad \forall\, \lambda\geq \lambda''.
\]
This ensures  that
$- \gv \left(\mx{I}+\frac{\G}{\lambda}\right)^k \G\1 >0 $ for $k=1,\dots,n$ where $n=\tau\lambda''$.
This is always possible when $\epsilon>0$.
\item Extend the $(\gv,\mx{G})$ representation with an Erlang tail of rate
$\lambda=\max(\lambda',\lambda'')$ and order $n=\lambda \tau$.
\end{itemize}

\item Step 5. If Step 2 was applied, at this point we have a Markovian representation for $f_Y$. To switch back to $f_X$, we use Lemma \ref{lemma:composition}. If Step 2 was not applied, Step 5 does not apply either.

\end{itemize}

\subsection{Step 1: Minimal representation}

{ Starting from representation $(\av,\A)$, we can obtain a minimal representation $(\av_1,\A_1)$
with the application of a representation minimization method.
A minimal representation can be obtained through several approaches. One possibility is directly from the pdf $f(x)=\av e^{\A x}\1$ as in Lemma \ref{lemma:odavissza}. Another, computationally stable order reduction method is the Staircase method from \cite{[BUCH10a]},
which uses singular value decomposition. In any case, the minimal representation $(\av_1,\A_1)$ enjoys properties P1-P5.}

There are two important properties that can be determined from a
minimal representation (or the density function directly).
These are the value and the  multiplicity of the dominant
eigenvalue and the validity of the dominant eigenvalue
condition. We denote the dominant eigenvalue (which is real
and negative) by $- \lambda_1$ and its multiplicity by $n_1$. Indeed,
$\lambda_1$ and $n_1$ determine the asymptotic rate of decay of the
pdf: it decays like $c_{\lambda_1, n_1}  x^{n_1-1} e^{-\lambda_1
x}$, where $c_{\lambda_1, n_1}$ is a positive constant, more precisely
\[ \lim_{x\rightarrow\infty} \frac{f(x)}{x^{n_1-1} e^{-\lambda_1 x}} = c_{\lambda_1, n_1}.
\]

\subsection{Step 2: Positive density at zero}
\label{sec:step3}

In the case when ME$(\av_1,\A_1)$ is such that $f_X(x)>0$ for positive $x$ values, but $f_X(0)=0$, then based on the following lemma,
we represent ME$(\av_1,\A_1)$ as the { convolution of an Erlang distribution and a matrix exponential distribution ME$(\av_2,\A_2)$
whose density is positive at $0$. Actually, it turns out from the proof of the following lemma that $\A_1=\A_2$.}

\begin{lemma}
\label{lemma:fnulla}
If $f_X(x)=-\av_1 e^{\A_1 x}\A_1\ones$ is a matrix exponential pdf with
\begin{equation}\label{eq:fx}
f_X(x)>0\quad\forall x>0,\quad \quad
\left. f_X^{(i)}(x)\right|_{x=0}=0 \quad i=0,\dots,l-1, \quad \quad \left. f_X^{(l)}(x)\right|_{x=0}>0,
\end{equation}
then $f_X$ can be written in the form
$$f_X=f_Y\ast g(l,\mu,\cdot),$$
for some large enough $\mu$, where
$g(l,\mu,x)=\frac{\mu^l x^{l-1}e^{-\mu x}}{(l-1)!}$ is the $\textrm{Erlang}(l,\mu)$ pdf,
$\ast$ denotes convolution and $f_Y(x)$ is a matrix exponential function with
$$f_Y(x)>0\quad\forall x\geq 0.$$
\end{lemma}

The proof of Lemma \ref{lemma:fnulla} is given in the Appendix. { The representation $(\av_2,\A_2)$ can be constructed either
from $f_Y$ via Lemma \ref{lemma:odavissza} or by using the fact that $\A_1=\A_2$ and calculating $\av_2$ from the appropriate linear equations.}

Lemma \ref{lemma:fnulla} and the following composition ensures that $f_X(x)$ and $f_Y(x)$
have a Markovian representation and satisfy the dominant eigenvalue condition
at the same time if $\mu > \lambda_1$.
\begin{lemma}
\label{lemma:composition}
If $f_Y(x)$ is ME distributed with representation $(\av_2,\A_2)$ of order $m$ and $\mu > \lambda_1$ then
$$f_X(x)=f_Y(x)\ast g(l,\mu,x),$$
is ME distributed with initial vector $\bv=\{1,0,0,\ldots, 0\}$ and generator matrix
\[\B=\left(
\begin{array}{cccc}
 -\mu & \mu    &        & \\
      & \ddots & \ddots &   \\
      &        & -\mu   & \mu \av_2 \\
      &        &        & \A_2  \\
\end{array}%
\right),
\]
where the first $l$ blocks of the matrix are of size one and the last block is of size $m$.
Additionally, if  $(\av_2,\A_2)$ is Markovian then $(\bv,\B)$ is Markovian as well.
\end{lemma}

\begin{proof}
Based on the structure of $\B$, the time to leave the first $l$ phases is $\textrm{Erlang}(l,\mu)$ distributed and the time spent in the set of phases from $l+1$ to $m$ is ME$(\av,\A)$ distributed.
\end{proof}

Based on Lemma \ref{lemma:fnulla} and \ref{lemma:composition} it remains to prove
that the matrix exponential density function $f(x)$ with $f(0)>0$ satisfying the dominant eigenvalue and the positive density conditions
has a Markovian representation.

\subsection{Step 3: Markovian generator}

\begin{figure}[t]
\centerline{
\includegraphics[width=0.8\textwidth]{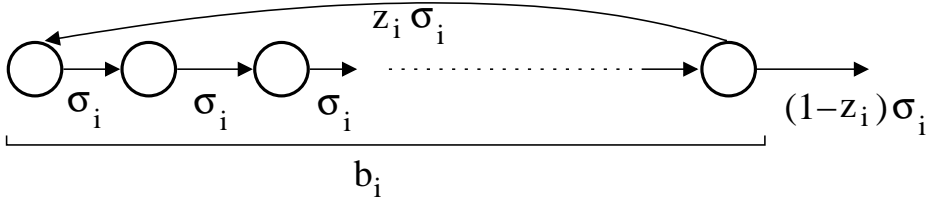}
}
\caption{FE-diagonal block.}
\label{fig:fe_block}
\end{figure}

\begin{figure}[t]
\centerline{
\includegraphics[width=0.8\textwidth]{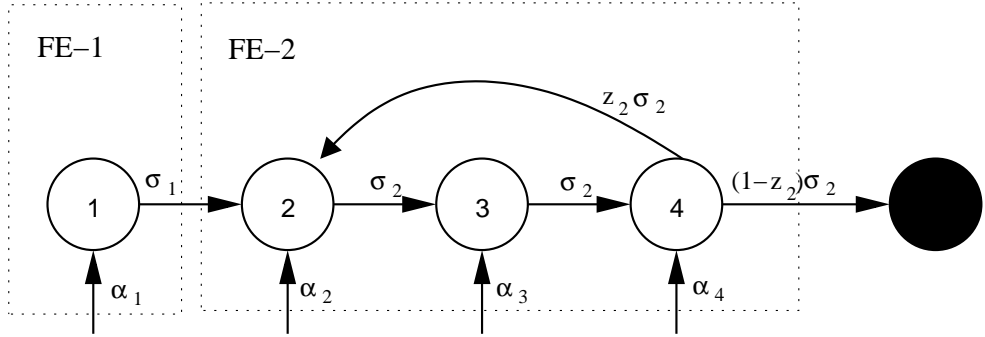}
}
\caption{FE-diagonal representation of a generator with a real eigenvalue ($\sigma_1$) and a pair of complex ones. }
\label{fig:fe_representation}
\end{figure}

The aim of this subsection is to transform the potentially non-Markovian representation $(\av,\mx{A})$ of
a ME distribution to a representation $(\gv,\mx{G})$ where $\mx{G}$ is a Markovian transient generator matrix
satisfying the properties of the matrix of a PH distribution (Definition \ref{def:ph}).
For matrix $\mx{G}$, we apply the matrix structure proposed in \cite{Mocanu99}.
It is a block bi-diagonal matrix structure, where each block represents
a real eigenvalue or a pair of complex conjugate eigenvalues of $\mx{A}$.
The blocks associated with real eigenvalue $-\lambda_i$ ($-\lambda_i<0$) are of size one,
the diagonal element is $-\lambda_i$
and the first sub-diagonal element is $\lambda_i$.
The blocks associated with complex eigenvalues are composed by Feedback-Erlang (FE)
blocks.
\begin{definition}\cite{Mocanu99}
A {\em Feedback-Erlang (FE) block} with parameters $(b, \sigma, z)$ is
a chain of $b$ states with transition rate $\sigma$ and one transition
from the $b$th state to the first state, with rate $z \sigma$ (c.f. \Figref{fig:fe_block}). The probability
$z\in[0,1)$ is called the feedback probability.
\end{definition}
A FE block $(b, \sigma, z)$ with length $b=1$ and $z=0$
corresponds to a real eigenvalue $-\sigma$ and is referred to as {\em degenerate} FE blocks.
Matrix $\mx{G}$ contains as many FE blocks (degenerate or non-degenerate) associated with a real eigenvalue or a pair of complex conjugate eigenvalues
as the multiplicity of the eigenvalue.
A non-degenerate FE block where $b$ is odd has a real eigenvalue and $(b-1)/2$ complex conjugate eigenvalue pairs.
A non-degenerate FE block where $b$ is even has 2 real eigenvalues and $(b-2)/2$ complex conjugate eigenvalue pairs.
In both cases the eigenvalues are equidistantly located on a circle in the complex plane around $-\sigma$.
The dominant eigenvalue of the FE block (the one with the largest real part) with parameters $(b, \sigma, z)$ is always real and
equals to $r = -\sigma\left(1-z^{\frac{1}{b}}\right)$~\cite{Mocanu99}.
Denote the eigenvalues of matrix $\mx{A}$ by $-\lambda_j$; the dominant eigenvalue (which is real) is $-\lambda_1$.
The FE blocks representing the eigenvalues are composed as follows
\begin{itemize}
\item if $\lambda_j$ is real, the corresponding FE block is a degenerate block; thus the parameters are:
\begin{align*}
\sigma_j=\lambda_j,~~~b_j=1,~~~z_j=0,
\end{align*}
\item if $\lambda_j=a_j\pm \ii c_j$ ($a_j>\lambda_1>0, c_j>0$) is a complex conjugate pair, the parameters are:
\begin{align*}
b_j&=\left\lceil \ds\frac{2 \pi}{\pi - 2 \arctan\left( \ds\frac{c_j}{-\lambda_1 + a_j} \right)}  \right\rceil,\\
\sigma_j&=\frac{1}{2}\left(- 2a_j-c_j\tan\frac{\pi}{b_j}+c_j\cot\frac{\pi}{b_j}\right),\\
z_j&=\left(1-\left(- a_j-c_j\tan\frac{\pi}{b_j}\right)/(2 \sigma_j) \right)^{b_j},
\end{align*}
where $\lceil x \rceil$ denotes the smallest integer { greater than or equal to $x$}.
\end{itemize}

This construction of the FE blocks ensures that $\lambda_1$ remains the dominant eigenvalue of matrix $\mx{G}$,
that is, the dominant eigenvalue of any FE block ($r_j$) is less than $-\lambda_1$ except the one(s) associated with $-\lambda_1$.

Connecting the obtained FE blocks such that the exit transition of an FE block (whose rate is $\lambda_j (1-z_j)$,
see \Figref{fig:fe_block}, in case of non-degenerate FE block and $\lambda_j$ in case of a degenerate one) is
connected to the first state of the next FE block composes a block bi-diagonal matrix
(c.f. \Figref{fig:fe_representation}). The obtained matrix $\mx{G}$ is Markovian and its Jordan form contains all Jordan blocks of
matrix $\mx{A}$. We order the FE blocks such that the first $n_1$ FE blocks are the $n_1$ degenerate FE blocks
associated with $-\lambda_1$.
The order of the rest of the FE blocks are irrelevant.
The FE blocks based finite Markovian representation of the eigenvalues of $\A$ is always feasible when the dominant eigenvalue condition holds.
If there was a pair of complex conjugate eigenvalues $a_j\pm \ii c_j$ which violates the dominant eigenvalue condition
such that $a_j=\lambda_1$ then the denominator of $b_j$ would be zero.

\Figref{fig:fe_representation} depicts an example of a Markovian generator which is the monocyclic
representation of a generator
with a dominant real eigenvalue ($-\lambda_1=-\sigma_1$) and a pair of complex conjugate eigenvalues in FE-diagonal
form. In this representation there are two FE blocks, one of length $b_1=1$ with
rate $\sigma_1$,
and one of length $b_2=3$ with rate $\sigma_2$ and feedback probability $z_2$. The associated generator matrix is
\[ \mx{G}=\left(%
\begin{array}{c|ccc}
  -\sigma_1 & \sigma_1 & 0 & 0 \\
  \hline
  0 & -\sigma_2 & \sigma_2 & 0 \\
  0 & 0 & -\sigma_2 & \sigma_2 \\
  0 & z \sigma_2 & 0 & -\sigma_2 \\
\end{array}%
\right)~.
\]

In order to find an equivalent representation of ME($\av, \mx{A}$) with  matrix $\mx{G}$ we need to compute
vector $\gv$, for which ME($\av, \mx{A}$) $\equiv$ ME($\gv, \mx{G}$),
with the help of Theorem \ref{thm:trafo}.
Let $n$ and $m$ ($n\leq m$) be the order of $\mx{A}$ and $\mx{G}$, respectively.
Compute matrix $\mx{\widehat{W}}$ of size $n\times m$ as the unique solution of
\begin{align*}
\mx{A}\mx{\widehat{W}}=\mx{\widehat{W}}\mx{G},~~~~~\mx{\widehat{W}}\1=\1,
\end{align*}
and based on $\mx{\widehat{W}}$  vector $\gv$ is
\begin{align*}
\gv = \av \mx{\widehat{W}}.
\end{align*}
Since $\mx{G}$ is Markovian, the obtained ($\gv, \mx{G}$) representation is Markovian if $\gv$ is non-negative,
but this is not necessarily the case. The case when $\gv$ has negative elements is considered in the following subsection.

\subsection{Step 4: Markovian vector}

At this point in the algorithm, the ME distribution is described by representation ($\gv, \mx{G}$) of order $u$
which has a block bi-diagonal, Markovian matrix $\G$, and a vector $\gv$
with at least one negative element. In the next step we extend the $(\gv,\G)$
representation with an additional $n$ phases in the following way.
\begin{equation}\label{eq:B}
\B=\left(
\begin{array}{cccc}
  \mx{G} & - \mx{G} \1     &         & \\
         & -\lambda & \lambda & \\
         &          & \ddots & \ddots  \\
         &          &        & -\lambda\\
\end{array}%
\right),
\end{equation}
where  $\B$ is of order $u+n$ (the size of the upper left block of
$\mx{B}$ is $u$, the remaining $n$ blocks are of size
one). $- \mx{G} \1$ is a non-negative column vector of size $u$. Due
to the structural properties of $\mx{G}$ it contains exactly one non-zero element,
which is the last element and it contains the exit rate from the last FE block.
The transformation matrix $\W$ of size $u \times (u+n)$, which transform from representation
($\gv, \mx{G}$) to representation ($\bv, \mx{B}$)
is the unique solution of $\G \W =\W \B$, $\W \1_{u+n} =\1_n$.
Fortunately, due to the special structure of matrix
$\mx{B}$, $\W$ is rather regular.

\begin{lemma} $\W$ has the following form:
\begin{align*}
\W=\left(%
\begin{array}{c|c|c|c|c}
  \left(\I+\frac{\G}{\lambda}\right)^n &  \left(\I+\frac{\G}{\lambda}\right)^{n-1}  \frac{-\G\1}{\lambda}
  &  \left(\I+\frac{\G}{\lambda}\right)^{n-2}  \frac{-\G\1}{\lambda}   & \ldots  &   \frac{-\G\1}{\lambda} \\
\end{array}%
\right),
\end{align*}
where the size of the first block is $u \times u$, the size of the remaining blocks is $1 \times u$.
\end{lemma}
\begin{proof}
Substituting this expression of $\W$ into $\G \W =\W \B$ and $\W \1_{u+n} =\1_n$ results {in} identities.
\end{proof}

Our goal is to find $n$ and $\lambda$ such that $\bv=\gv \W$ is Markovian (that is non-negative), where
\begin{align}
\label{Walak}
\gv \mx{W}=\left(%
\begin{array}{c|c|c|c|c}
 \gv  \left(\I+\frac{\G}{\lambda}\right)^n &  \gv  \left(\I+\frac{\G}{\lambda}\right)^{n-1}\frac{-\G\1}{\lambda}
 & \gv  \left(\I+\frac{\G}{\lambda}\right)^{n-2}  \frac{-\G\1}{\lambda} & \ldots  & \gv ~\frac{-\G\1}{\lambda} \\
\end{array}%
\right).
\end{align}
The first block of this vector is of size $u$ and the remaining $n$ blocks are of size $1$.
We need to prove that this vector is nonnegative for an appropriate pair $(\lambda,n)$.

\begin{theorem}
\label{thm:lambdan}
There exists a pair $(\lambda,n)$ such that $\gv\W$ is strictly positive.
\end{theorem}

The rest of this subsection  is devoted to proving Theorem
\ref{thm:lambdan}. We assume everything that was done so far, for example\
that the dominant eigenvalue condition and the positive density
condition hold, the density is positive at zero and also that the matrix $\mx{G}$ is Markovian and in
monocyclic form such that the degenerate FE block(s) representing the dominant eigenvalue $-\lambda_1$ are the first one(s).
First we present a heuristic argument, then the formal proof.

\subsubsection{Heuristic argument}

$\lambda$ and $n$ are  typically chosen to be large (see
\cite{Mocanu99}). However, finding an appropriate pair is not as
simple as choosing some large $\lambda$ and a large $n$. For each
$n$, the set of appropriate values of $\lambda$ forms a finite
interval. If $n$ is large enough, this interval is nonempty, but --
without further considerations -- it is impossible to identify this
interval (or even one element of it). Vice versa, for each $\lambda$
there is a finite set of appropriate values for $n$. This means that
the naive algorithm of increasing the values of $n$ and $\lambda$ --
without further considerations -- may possibly never yield an
appropriate pair. For this reason, we instead propose a different
parametrization, which takes the dependence between $n$ and
$\lambda$ into account better.

Let $\tau=n/\lambda$.  $\tau$ turns out to be a value interesting in
its own right. The ME pdf resulting from the pair $(\gv\W,\B)$ has a
term coming from the first block of $\B$ and it has $n$ terms coming
from the Erlang-tail. We argue that the terms coming from the
Erlang-tail can be regarded as an approximation of the original pdf
on the interval $[0,\tau]$, while the term coming from the first
block is some sort of correction that makes the approximation
exactly equal to the original pdf. Each of the terms in the
Erlang-tail contribute an Erlang pdf with rate $\lambda$ and order
$k\in [1,\dots, n]$ to the pdf. The Erlang($\lambda,k$) pdf is
concentrated around the point $\frac{k}{\lambda}=\frac{k \tau}{n}$. These
points are situated along the interval $[0,\tau]$ in an equidistant
way with distance $\frac{1}{\lambda}$.

The weight (initial probability) of the Erlang pdf centered around the point $\frac{k}{n}\tau$ is
$$\gv  \left(\I+\frac{\G}{\lambda}\right)^k  \frac{-\G\1}{\lambda} \approx
\gv  e^{\frac{k\tau}{n}\G}
 \frac{-\G\1}{\lambda}=\frac{1}{\lambda}f_X\left(\frac{k\tau}{n}\right),$$
which means that the weights are approximately equal to samples of the original pdf
at points $\frac{k\tau}{n}$, $k\in [1,\dots, n]$ divided by $\lambda$, resulting in a pdf that is approximately
equal to the original along the interval $[0,\tau]$.

The first block of  $\gv \W$ is different. From the form of $\B$ it
is clear that the contribution of the first block is concentrated
after the point $\tau$; the role of this block is essentially to make
a correction in the interval $[\tau,\infty]$, where the previous
Erlang-approximation does not hold.

Altogether the previous argument can be depicted nicely in Figures \ref{fig:appr_pdf} and \ref{fig:neg_pdf}. We denote
$$f_k(x)=\gv  \left(\I+\frac{\G}{\lambda}\right)^k  \frac{-\G\1}{\lambda} ~~g(k,\lambda,x) ,\quad k=0,\dots, n-1$$
the approximating Erlang terms and
$$f_0(x)=\gv  \left(\I+\frac{\G}{\lambda}\right)^n e^{-\G x} (-\G \1) \ast g(n,\lambda,x)$$
the correction term.  In Figure \ref{fig:appr_pdf}, the
approximating Erlang terms roughly follow the graph of $f_X$, while
$f_0$ is concentrated after $\tau$. (The values are $\tau=3,
\lambda=12$ and $n=36$; to make the figure visually apprehensible,
only some of the approximating Erlang functions were included with
slightly increased weights.)

\begin{figure}[t]
\centerline{
\includegraphics[width=0.8\textwidth]{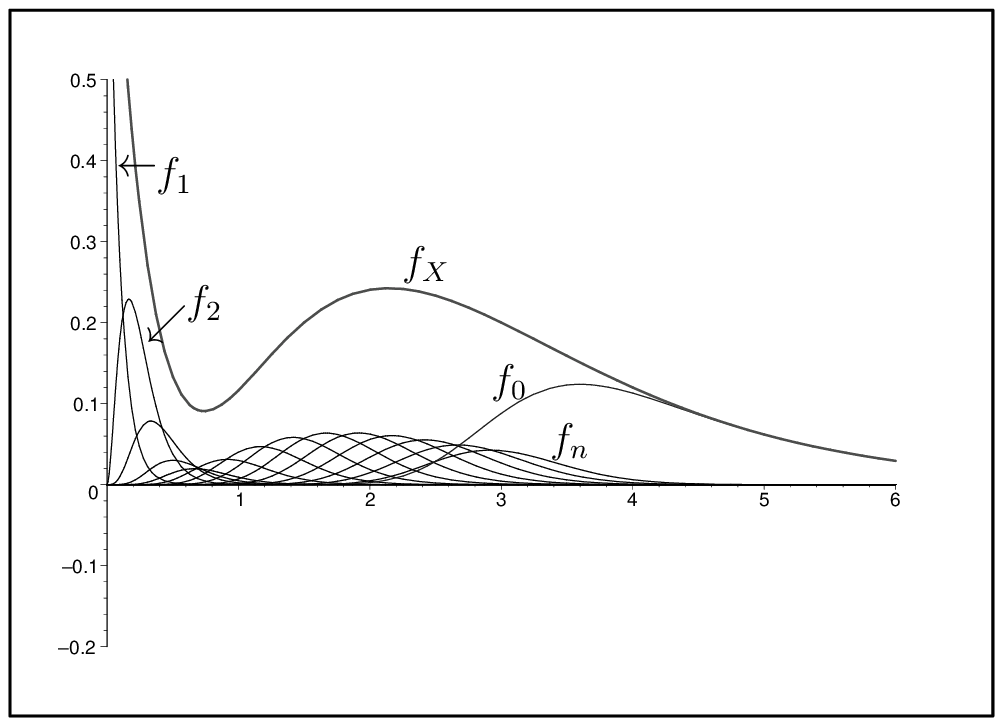}
}
\caption{Erlang pdf's approximating the original one}
\label{fig:appr_pdf}
\end{figure}

The value of $\lambda$  controls how concentrated the approximating
Erlang pdf's are and also controls how close their weights are to
the sampling of the original pdf. Given that $f_X(x)>0$ for $x>0$,
this means that for \emph{any} choice of $\tau$, the
Erlang-approximation has positive weights if $\lambda$ is large
enough. The choice of $\tau$ is only important to make sure that the
weights assigned to the correction term are also positive. Figure
\ref{fig:neg_pdf} shows an example where $\lambda$ is too small
(notably $\lambda=4$). In this case, some of the approximating
Erlang functions have negative coefficients.

\begin{figure}[t]
\centerline{
\includegraphics[width=0.8\textwidth]{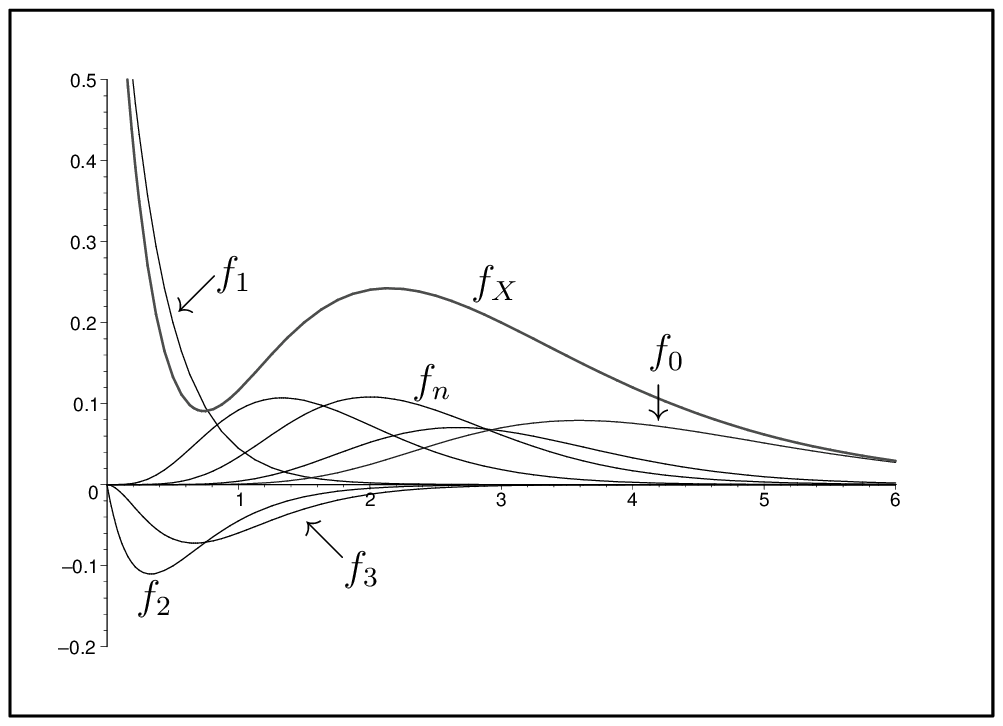}
}
\caption{If $\lambda$ is too small, some Erlang pdf's are negative}
\label{fig:neg_pdf}
\end{figure}

\subsubsection{Formal proof}

Before the actual proof, some results are stated as standalone lemmas. Their proofs are in Appendix \ref{sec:appc}.

The first one is essentially a real approximation, so we state it in
that form too,  along with the matrix version which is useful for
our purposes. Relevant properties of matrix (and vector) norms can
be found in Appendix \ref{sec:appa}.

\begin{lemma}
\label{lemma:kozelites}

\begin{enumerate}
\item[i)] For any fixed $r>0$ and positive integer $n$,
$$\sup_{|z|\leq r}\left|e^z-\left(1+\frac{z}{n}\right)^n\right|\leq   \frac{r^2e^{r}}{2n},$$
and the supremum is obtained at $z=r$.
\item[ii)] For any $\H$ square matrix,
$$\left\| {e^{\H}} - \left(\I+\frac{\H}{n}\right)^n\right\| \leq  \frac{r^2e^{r}}{2n},$$
where $r=\|\H\|$.
\end{enumerate}
\end{lemma}

We state one more lemma. It identifies the main terms in $e^{\G x}$ when $\G$ is in monocyclic form.

\begin{lemma}
\label{lemma:csillag}
\begin{align*}
& \left(e^{\G x}\right)_{1j} \sim C_j x^{j-1}e^{-\lambda_1 x} \quad \textrm{ if } 1  \leq j \leq n_1\\
& \left(e^{\G x}\right)_{1j} \sim C_j x^{n_1-1}e^{-\lambda_1 x} \quad \textrm{ if } n_1 <  j \leq u , \\
& \lim_{x\to\infty}\frac{\left(e^{\G x}\right)_{ij}}{\left(e^{\G x}\right)_{1j}}=0\quad \textrm{ if } \quad 2\leq i\leq u, ~~ 1 \leq j \leq u ,
\end{align*}
where $C_j$ denote positive (combinatorial) constants and $f(x) \sim g(x)$ denotes that $\lim_{x\to\infty} f(x)/g(x)=1$.
The last relation means that the first row dominates all other rows as $t$ tends to infinity.
\end{lemma}

Note that the last part of Lemma \ref{lemma:csillag} is stated as
$\frac{\left(e^{\G x}\right)_{ij}}{\left(e^{\G x}\right)_{1j}}\to
0$; in fact, the elements $\left(e^{\G x}\right)_{ij}$ are in a form
similar to $\left(e^{\G x}\right)_{1j}$, just with either the
same exponential term and lower degree polynomial terms,  or lower
exponent (and in this case, the polynomial term does not matter).
The actual exponents and polynomial terms, along with the constants
$C_j$ can be calculated explicitly from the proof of Lemma
\ref{lemma:csillag}, but will not be used.

{ We emphasize that Lemma \ref{lemma:csillag} relies heavily on the monocyclic structure of $\G$, notably on the fact that the upper bi-diagonal elements (elements $(1,2),(2,3),\dots$) of the matrix are strictly positive.}

Now we are ready to prove Theorem \ref{thm:lambdan}.

\subsubsection*{Proof of Theorem \ref{thm:lambdan}.}

We assume that the matrix exponential density function $f_X$ associated with representation $(\gv,\G)$
satisfies $f_X(0)>0$, the dominant eigenvalue and the positive density conditions, and that $\G$ is in monocyclic block structure with the first block corresponding to the dominant eigenvalue $\lambda_1$.

\medskip \noindent
{\em First we show that the first coordinate of $\gv$, denoted by $\gv_1$, is positive.}

If $\gv_1=0$, then the multiplicity of $-\lambda_1$ is $n_1-1$
according to the structure of matrix $\G$ (see \eqref{eq:strG} in the proof of Lemma \ref{lemma:csillag} in subsection \ref{subsection:csillag}), which is in conflict with the fact that the
multiplicity of $-\lambda_1$ in the minimal representation is $n_1$.

$f_X(x)$ is dominated by the first row of $e^{\G x}$ for large values of $x$ and consequently the sign of $f_X(x)$ is determined by $\gv_1$.
The elements of $e^{\G x}$ are transient probabilities of the Markov chain with generator $\G$, consequently they
are non-negative. The elements of the first row of $e^{\G x}$ are strictly positive for $x>0$ because the FE-blocks are connected that way that all
states are reachable from the first state (cf. Figure \ref{fig:fe_representation}).
According to Lemma \ref{lemma:csillag}
$f_X(x)$ is dominated by the first row of $e^{\G x}$ for large values of $t$ and consequently the sign of $f_X(x)$ is determined by $\gv_1$.
More precisely, Lemma \ref{lemma:csillag} implies that
$$0<f_X(x)=\gv (-\G) e^{\G x}\1\sim C \lambda_1 \gv_1 x^{n_1-1}e^{-\lambda_1 x}$$
where $C=\sum_{j\geq n_1} C_j>0$ and $\lambda_1>0$.

\medskip \noindent
{\em Next we show that there exists a $\tau$ such that $\gv e^{\G \tau}$ is positive.}

For the first row of $\gv e^{\G x}$ we have
\begin{align*}
\left(\gv e^{\G x}\right)_{1j} &\sim C_j \gv_1 x^{j-1}e^{-\lambda_1 x} & \textrm{ if } j < n_1, \\
\left(\gv e^{\G x}\right)_{1j} &\sim C_j \gv_1 x^{n_1-1}e^{-\lambda_1 x} & \textrm{ if } n_1 \leq j \leq u,
\end{align*}
from Lemma \ref{lemma:csillag}.
Thus $\gv e^{\G x}$ is positive if $x$ is large enough.
For a constructive procedure to find $\tau$, one can double $x$ starting from $n_1/\lambda_1$ as long as $\min(\gv e^{\G x})<0$.
It is not necessary to find the smallest $x$ for which $\gv e^{\G x}$ is nonnegative.

\medskip \noindent
{\em
After that we show that there exists $\lambda'$ such that $\gv (\I+\frac{\G}{\lambda})^{\lambda \tau}>0$ for $\lambda\geq \lambda'$.}

Apply Lemma \ref{lemma:kozelites} with $\H=\G \tau$ and $n=\lambda \tau$ to get that
$$\left\| \left(\I+\frac{\G}{\lambda}\right)^{\lambda \tau}- e^{\G \tau}\right\|\to 0$$
as $\lambda\to\infty$, and consequently
$$\left\|\gv \left(\I+\frac{\G}{\lambda}\right)^{\lambda \tau}-\gv e^{\G \tau}\right\|\to 0,$$
meaning that $\gv (\I+\frac{\G}{\lambda})^{\lambda \tau}$ is also strictly
positive if $\lambda$ is large enough.
Let $\epsilon_1=\min(\gv e^{\G \tau})$; in accordance with Lemma \ref{lemma:kozelites}, define $\lambda'$ as the solution of
\begin{equation}\label{eq:lambda'}
\|\gv\| \frac{(g\tau)^2e^{g\tau}}{2\lambda\tau}=\epsilon_1.
\end{equation}
where $g=\|\G\|$. Then $\gv (\I+\frac{\G}{\lambda})^{\lambda \tau} > 0$ for $\lambda > \lambda'$,
because the left-hand side is a strictly monotone decreasing function of $\lambda$. Note that $\lambda'$
is explicitly computable from \eqref{eq:lambda'}.

Next we investigate the sign of the rest of the elements of vector $\gv\W$. We apply Lemma \ref{lemma:kozelites} again,
this time for $\H=\frac{k\G}{\lambda}$ and $n=k$ to get
$$\left\|e^{\frac{k\G}{\lambda}}-\left(\I+\frac{\G}{\lambda}\right)^k
\right\| \leq
e^{\frac{k g}{\lambda}} \frac{(kg)^2}{2k\lambda^2}\leq e^{\tau g} \frac{\tau g^2}{2\lambda}
$$
uniformly in $0\leq k \leq \lambda\tau$.

Let $\epsilon_2=\inf_{0\leq x\leq \tau} f_X(x) = \inf_{0\leq x\leq \tau}\gv e^{\G x} (-\G) \1$.
Since $f_X(0)>0$ as a result of Step 3 in Section \ref{sec:step3}, $\epsilon_2$ is strictly positive, due to the positive density condition.
Let $V_k$ be the $k$-th coordinate of $\gv \W$ associated with the Erlang tail in \eqref{Walak}; that is,
$$V_k=\gv  \left(\I+\frac{\G}{\lambda}\right)^k\frac{-\G\1}{\lambda}.$$

Then
\[ \left| \lambda V_k - f_X(\textstyle\frac{k}{\lambda}) \right| =
\left|\gv \left[ e^{\frac{k\G}{\lambda}} - \left(\I+\frac{\G}{\lambda}\right)^k \right] \G \1 \right| \leq
\|\gv\| \left\| e^{\frac{k\G}{\lambda}} - \left(\I+\frac{\G}{\lambda}\right)^k \right\| \|\G\| \|\1\| \leq
\|\gv \| e^{\tau g} \frac{\tau g^2}{2\lambda}g\|\1\|.
\]
Define $\lambda''$ as the solution of
\begin{equation}\label{eq:lambda''}
\|\gv \| e^{\tau g} \frac{\tau g^2}{2\lambda}g\|\1\|= \epsilon_2.
\end{equation}
$\lambda''$ is also explicitly computable. (Note that $\|\1\|=1$, see Appendix \ref{sec:appa}).
For all $\lambda>\lambda''$ we have $V_k>0$ because
$f_X(\textstyle\frac{k}{\lambda})\geq \epsilon_2$ and the difference between
$\lambda V_k$ and $f_X(\textstyle\frac{i}{\lambda})$ is less than $\epsilon_2$.

Putting these together, we get that for $\tau$ and $\lambda=\max(\lambda',\lambda'')$ both
parts of the vector $\gv\W$, that is,
$\gv  \left(\I+\frac{\G}{\lambda}\right)^n $ and $  \gv  \left(\I+\frac{\G}{\lambda}\right)^{k}\frac{-\G\1}{\lambda}$
for $k=0,1,\ldots,n-1$,
are positive where $n=\lceil\tau \lambda \rceil$ and the obtained representation is indeed Markovian.

\subsection{Step 5: correction related to Step 2}

If Step 2 was applied, $(\bv,\B)$ is actually a Markovian representation for $f_Y$; Lemma \ref{lemma:composition} ensures that
\begin{align*}
\bv'&=\{1,0,0,\ldots, 0\}\\
\B'&=\left(
\begin{array}{cccc}
 -\mu & \mu    &        & \\
      & \ddots & \ddots &   \\
      &        & -\mu   & \mu \bv \\
      &        &        & \B  \\
\end{array}
\right)
\end{align*}
is a Markovian representation for $f_X(x)=f_Y(x)\ast g(l,\mu,x)$.

{
\section{Worked example}

Let
\begin{align*}
\av=&\frac{102}{139}\left(\begin{array}{ccccccc}
1 & 1 & -\frac13 & -\frac{2}{3} & -\frac{5}{2} & \frac{12}{17} & \frac{14}{17}
\end{array}
\right),
\\
\A=&\left(
\begin{array}{ccccccc}
-1 & 1 & 0 & 0 & 0 & 0 & 0\\
0 & -1 & 0 & 0 & 0 & 0 & 0\\
0 & 0 & -1 & 4 & 0 & 0 & 0\\
0 & 0 & 1 & -1 & 0 & 0 & 0\\
0 & 0 & 0 & 0 & -4 & 0 & 0\\
0 & 0 & 0 & 0 & 0 & -5 & 3\\
0 & 0 & 0 & 0 & 0 & -3 & -5\\
\end{array}
\right),
\end{align*}
then
$$f(x)=- \av \A e^{\A x} \1=\frac{102}{139}\left(xe^{-x}+e^{-x}+e^{-3x}-10e^{-4x}+e^{-5x}\left(8\cos(3x)+4\sin(3x)\right)\right).$$
The eigenvalues of $A$ are $-1$ (with multiplicity 2), $-3, -4, -5+3i, -5-3i$ and $1$. The eigenvalue 1 is redundant: the corresponding
right-eigenvector is orthogonal to $\av$, thus it does not appear in the pdf. It is eliminated during Step 1.

After Step 1, a minimal representation is obtained:
\begin{align*}
\av_1=&\frac{102}{139}\left(\begin{array}{cccccc}
1 & 1 & \frac13 & -\frac{5}{2} & \frac{13+\ii}{17} & \frac{13-\ii}{17}
\end{array}
\right),
\\
\A_1=&\left(
\begin{array}{cccccc}
-1 & 1 & 0 & 0 & 0 & 0\\
0 & -1 & 0 & 0 & 0 & 0\\
0 & 0 & -3 & 0 & 0 & 0\\
0 & 0 & 0 & -4 & 0 & 0\\
0 & 0 & 0 & 0 & -5+3\ii & 0\\
0 & 0 & 0 & 0 & 0 & -5-3\ii\\
\end{array}
\right).
\end{align*}

Since $f(0)=0$, Step 2 needs to be applied.
\begin{align*}
f(0)=0\qquad f'(0)=7>0,
\end{align*}
so the value of $k$ in Lemma \ref{lemma:fnulla} is $k=1$. Setting $\mu=10$, the transformed pdf after Step 2 (borrowing the notation $f_Y$ from Lemma \ref{lemma:fnulla}) is
$$f_Y(x)=\frac{102}{139}\left(
\frac{9}{10}xe^{-x}+e^{-x}+
\frac{7}{10}e^{-3x}-6e^{-4x}+
\frac{13+\ii}{5}e^{(-5+3\ii)x}+
\frac{13-\ii}{5}e^{(-5-3\ii)x}
\right)$$
and the corresponding representation for $f_Y$ is
\begin{align*}
\av_2=&\frac{102}{139}\left(\begin{array}{cccccc}
\frac{9}{10} & 1 & \frac{2}{5} & -\frac{7}{15} & \frac{23}{340} & \frac{7}{340}
\end{array}
\right),
\\
\A_2=&\left(
\begin{array}{cccccc}
-1 & 1 & 0 & 0 & 0 & 0\\
0 & -1 & 0 & 0 & 0 & 0\\
0 & 0 & -3 & 0 & 0 & 0\\
0 & 0 & 0 & -4 & 0 & 0\\
0 & 0 & 0 & 0 & -5+3\ii & 0\\
0 & 0 & 0 & 0 & 0 & -5-3\ii\\
\end{array}
\right).
\end{align*}

From now on, we work with this representation. In Step 3, the eigenvalue pair $5\pm 3\ii$ is represented by a feedback-Erlang block. The order of this pair is $b=4$, and the corresponding FE-block is
\begin{align*}
&\left(
\begin{array}{cccc}
-5 & 5 & 0 & 0\\
0 & -5 & 5 & 0\\
0 & 0 & -5 & 5\\
\frac{81}{125} & 0 & 0 & -5\\
\end{array}
\right).
\end{align*}
Step 3 results in the representation
\begin{align*}
\gv=&\frac{102}{139}\left(\begin{array}{cccccccc}
\frac{315}{2176} &
\frac{10733}{21760} &
\frac{6641}{32640} &
\frac{8399}{21760} &
\frac{147}{680} &
-\frac{67}{272} &
-\frac{45}{1088} &
\frac{225}{1088}
\end{array}
\right),
\\
\G=&\left(
\begin{array}{cccccccc}
-1 & 1 & 0 & 0 & 0 & 0 & 0 & 0\\
0 & -1 & 1 & 0 & 0 & 0 & 0 & 0\\
0 & 0 & -3 & 3 & 0 & 0 & 0 & 0\\
0 & 0 & 0 & -4 & 4 & 0 & 0 & 0\\
0 & 0 & 0 & 0 & -5 & 5 & 0 & 0\\
0 & 0 & 0 & 0 & 0 & -5 & 5 & 0\\
0 & 0 & 0 & 0 & 0 & 0 & -5 & 5\\
0 & 0 & 0 & 0 & \frac{81}{125} & 0 & 0 & -5\\
\end{array}
\right).
\end{align*}

Since $\gv$ still contains negative elements, Step 4 needs to be applied.

Following the algorithm in the proof of Theorem \ref{thm:lambdan}, we obtain the following values:
\begin{itemize}
\item $\tau=0.5$ (from $\gv e^{\G\tau}>0$),
\item $g=\|G\|=10$,
\item $\|\gv\|<1.5$,
\item $\epsilon_1>0.05$ (for $\tau=0.5$),
\item $\lambda'=112000$ from \eqref{eq:lambda'},
\item $\epsilon>0.069$, and thus $\lambda''=806600$ from \eqref{eq:lambda''}.
\end{itemize}
This means that applying Step 4 with $\lambda=806600$ and $n=\tau\lambda=403300$ we obtain a Markovian representation for $f_Y$ in the form
of \eqref{eq:B}.

Finally, Step 5 applies, so by Lemma \ref{lemma:composition} with $\mu=10$ and $k=1$, we obtain a Markovian representation for the
original ME($\av,\A$). The representation is of order $403309$. Note that the order of this representation is \emph{very} far
from minimal, but we do not pursue a minimal value.
}

\section{Conclusion}

We have proposed a constructive proof for O'Cinneide's characterization theorem \cite{OCIN90} along with an algorithm that
always succeeds in finding a Markovian representation. The algorithm and the proof are divided into a few distinct steps,
connecting some of the modern results in the field as well as introducing some original ideas using elementary function theory and matrix analysis.

\section*{Acknowledgement}
I. Horv\'ath was supported by the Hungarian National Science
Foundation, OTKA, grant K100473 and by the
Hungarian Government through the project
T\'AMOP-4.2.2.B-10/1--2010-0009;
and M. Telek was supported by OTKA grant K101150
and T\'AMOP-4.2.2C-11/1/KONV-2012-0001.

\bibliographystyle{plain}
\bibliography{procedure}

\appendix
\section{Vector and matrix norms}
\label{sec:appa}

We need some auxiliary facts about vector and matrix norms. First we define vector norms. Let $\v$ be a vector of size $n$.

\begin{definition}
The 1-norm and $\infty$-norm of $\v$ are
$$\|\v\|_1=\sum_{i=1}^n |v_i|,\qquad \|\v\|_\infty=\max_{1\leq i\leq n} |v_i|.$$
\end{definition}

\begin{lemma}
\begin{enumerate}
\item[a)] $\|.\|_1$ and $\|.\|_\infty$ are equivalent, i.e,
$$\|\v\|_\infty\leq\|\v\|_1\leq n ~ \|\v\|_\infty.$$
\item[b)] If $\v$ is a row vector and $\w$ a column vector, then
$$|\v\w|\leq \|\v\|_1 ~ \|\w\|_\infty,\qquad |\v\w|\leq \|\v\|_\infty ~ \|\w\|_1.$$
\end{enumerate}
\end{lemma}

The fact that they are equivalent means that they define the same
topology,  so convergence to $0$ is equivalent in either norm. For
convenience, we will stick to using $\|.\|_1$ for row vectors and
$\|.\|_\infty$ for column vectors.

We also need a matrix norm.
\begin{definition}
The $\infty$-norm of $\A$ is
$$\|\A\|_\infty=\max_{1\leq i\leq n} \sum_{j=1}^n |A_{ij}|$$
\end{definition}

This is a submultiplicative norm:
$$\|\A\B\|_\infty\leq \|\A\|_\infty ~ \|\B\|_\infty.$$

Actually, the above matrix norm is the induced matrix norm of the
vector  norm $\|.\|_\infty$ when multiplying a column vector with a
matrix from the left, and the induced matrix norm of the vector norm $\|.\|_1$ when multiplying a row vector with a matrix from the right.
This means it works nicely with the previous vector norms.
\begin{lemma}
Let $\v$ be a row vector and $\w$ a column vector of size $n$ and $\A$ be an $n\times n$ matrix. Then
$$\|\v\A\|_1\leq \|v\|_1 ~ \|\A\|_\infty,\quad
\|\A\w\|_\infty\leq \|\A\|_\infty ~ \|\w\|_\infty,\quad
|\v\A\w|\leq \|\v\|_1 ~  \|\A\|_\infty ~  \|\w\|_\infty.$$
\end{lemma}

\section{Proofs for the necessary direction}
\label{sec:appb}

\begin{definition}
\label{def:redundant}
The Markovian ($\av$, $\A$) representation of PH($\av$, $\A$) is {\em redundant} if it contains at least one state
which cannot be visited by the Markov chain with initial distribution $\av$ and generator $\A$. Otherwise ($\av$, $\A$) is {\em non-redundant}.
\end{definition}

If the representation ($\av$, $\A$) is redundant then it is easy to identify and eliminate the redundant states.
Consider the vector $-\av\A^{-1}$. The stochastic interpretation of its $i$th coordinate is the mean time spent in state $i$ before absorption.
If the $i$th element of vector $-\av\A^{-1}$ is zero then state $i$ is redundant and the associated elements can be deleted from vector
$\av$ and matrix  $\A$ without modifying the distribution of time till absorption.

\begin{lemma}
\label{lemma:pdc}
If $X$ is PH($\av$,$\A$) distributed and non-redundant, then the positive density condition holds, that is,
$$f_X(x)>0\qquad \forall x>0.$$
\end{lemma}

\begin{proof}
If $X$ is PH($\av$,$\A$) distributed  and non-redundant
then there is a path from every state with positive initial probability to the absorbing state and every state belongs to one of those paths.
Consequently, the Markov chain is in state $j$ at time $x$ with positive probability, for any time $x>0$ and for any state $j$.
Let state $i$ be a transient state from where the absorption rate $g_i$ is positive. Then
\[ f_X(x)= \av  e^{\A x} (-\A)\1 = \sum_{j=1}^n \mathrm{Pr}(Z(x)=j) g_j \geq Pr(Z(x)=i) g_i > 0, \]
where $Z(x)$ denotes the underlying Markov chain.
\end{proof}

\begin{lemma}
\label{lemma:dec}
If $X$ is PH($\av$,$\A$) distributed and non-redundant,
then the dominant eigenvalue condition holds.
\end{lemma}

Before proving Lemma \ref{lemma:dec}, we elaborate on Definition \ref{def:minimal}. Let $ME(\gv,\G)$ be a minimal representation for $X$.
Consider its pdf using the Jordan-decomposition of $\G$ ($\G=\P\J \P^{-1}$)
$$f_X(x)=-\gv \P\J e^{\J x}\P^{-1}\1 =\sum_{i=1}^l-\gv \P_i\J_i e^{\J_ix}\P_i'\1,
$$
where $\J_i$ denotes the Jordan-block corresponding to the eigenvalue $-\lambda_i$ and $\P_i$ denotes the submatrix of $\P$ containing
only the columns corresponding to $\J_i$. $\P_i'$ denotes the submatrix of $\P^{-1}$ that contains only the rows corresponding to $\J_i$
(thus $\P_i$ is of size $n\times n_i$, where $n_i$ is the multiplicity of $-\lambda_i$ and $n$ is the size of $\G$,
and $\P_i'$ is of size $n_i\times n$). In $\P_i$, the first column of each block is the (unique, up to a constant factor) right eigenvector $\v_i$
corresponding to that eigenvalue and the other columns are generalized eigenvectors. Similarly in $\P_i'$, the last row of each
block is the (unique, up to a constant factor)
left eigenvector $\u_i$ corresponding to that eigenvalue and the rest of the rows are generalized eigenvectors.
If $i\neq j$, then $\P_i' \P_j=\mx{0}$.

The dominant term of $e^{\J_ix}$ is equal to $\frac{x^{n_i-1}e^{-\lambda_i x}}{(n_i-1)!}$ (where $n_i$ denotes the size of $\J_i$),
and it is situated in the upper right corner. Within $-\gv \P_i\J_ie^{\J_ix}\P_i'\1$ this dominant term is obtained exactly when taking
$$-\gv \v_i \J_ie^{\J_ix}\u_i \1 = (\gv \v_i)\lambda_i\frac{x^{n_i-1}e^{-\lambda_i x}}{(n_i-1)!} (\u_i \1).$$
If any of the coefficients $(\gv \v_i)$ and $(\u_i \1)$ is 0, this term would vanish. Properties P3 and P4 ensure that this is not the case,
in other words, all eigenvalues contribute to the pdf with maximal multiplicity (that is, Property P2).

This allows us to prove the DEC for any (possibly non-minimal) Markovian representation $(\av,\mx{A})$ by proving
that there exists a real eigenvalue of $\mx{A}$ that is strictly greater than the real part of all other eigenvalues AND this eigenvalue
contributes to the pdf with maximal multiplicity.

The proof of Lemma \ref{lemma:dec} is based essentially on the Perron--Frobenius lemma.  We
begin by citing the Perron--Frobenius lemma along with a necessary
definition, see for example\ \cite{Meye04}.

\begin{definition}
An $n\times n$ matrix $\mx{M}$ is \emph{reducible} if there exists a nontrivial partition $I\cup J$ of $\{1,2,\dots,n\}$ such that
$$\M_{ij}=0\qquad \forall i\in I, j\in J.$$

Otherwise, $\M$ is \emph{irreducible}.
\end{definition}

In case $\mx{M}$ is the transient generator of a PH distribution, then irreducibility means that each state can be reached from any
other state before absorption, in this case we say that $\mx{M}$ has a single communicating class.
If the Markov chain defined by $\mx{M}$ has multiple
communicating classes, they correspond to a partition of the states as in the above definition.

\begin{theorem}[Perron--Frobenius]
\label{thm:pf} If the irreducible matrix $\M$ has nonnegative
elements, then there exists a positive eigenvalue $\nu_1$ of
$\mx{M}$ such that
\begin{itemize}
\item $\nu_1$ has multiplicity $1$,
\item $\nu_1\geq |\nu_i| \,\forall i$ where $v_i$ denote the eigenvalues of $\M$, and
\item the corresponding right-eigenvector $\v_1$ is strictly positive (note that $\v_1$ is unique up to a constant factor;
it can be chosen such that $\v_1$ is strictly positive).
\end{itemize}
\end{theorem}

See Theorem 3 in \cite{suzumura} for a short, self-contained proof or Chapter 8 in \cite{Meye04} for a more detailed discussion.
Note that the same conclusion holds for the left-eigenvector $\u_1$ as well. Note that the fact that $\nu_1$ is positive with multiplicity 1
and $\nu_1\geq |\nu_i|$ mean that $\Re(\nu_i)<\nu_1$ for $i\neq 1$.

\subsubsection*{Proof of Lemma \ref{lemma:dec}.}

In case $\A$ has a single communicating class we apply Theorem
\ref{thm:pf} to the matrix $\mx{M}=\A+\omega \I$, where $\omega=\max_i |a_{ii}|$. Given
that the matrix $\A$ is Markovian, $\mx{M}$ is nonnegative with the
same eigenvectors and the eigenvalues shifted by $\omega$.
The dominant eigenvalue $\nu_1$ of $\mx{M}$ corresponds to
the dominant eigenvalue
$-\lambda_1$ of $\A$, that is $\nu_1=-\lambda_1+\omega$ and
the same relation holds for the other eigenvectors. Clearly for $i\neq 1$
$$\Re(\nu_i)<\nu_1 \quad \Longrightarrow \quad \Re(-\lambda_i)<-\lambda_1.$$

If $\mx{A}$ has a single communicating class
then Theorem \ref{thm:pf} guarantees that the multiplicity of $-\lambda_1$ is 1; this means that the unique
dominant term in the pdf is
$(\av \v_1)\lambda_1e^{-\lambda_1 x}(\u_1 \1).$
Strict positivity of $\v_1$ and $\u_1$ ensure $\av \v_1>0$ and $\u_1 \1>0$, so indeed $\lambda_1$ contributes to the
pdf with multiplicity 1.

If $\A$ has several communicating classes, the states can be renumbered such that $\A$ is an upper block triangular matrix,
where each diagonal block corresponds to a communicating class and the blocks above the diagonal correspond to transitions between classes.
The diagonal blocks are denoted by $\B_1,\dots, \B_k$. The eigenvalues of $\A$ are the union of the eigenvalues associated with these diagonal blocks.
Each $\B_i$ is itself the generator of a transient Markov chain, and, since $\B_i$ is also irreducible,
Theorem \ref{thm:pf} can be applied to each of them. It follows that each of these blocks (communicating classes) has its own dominant eigenvalue
such that within that class, the real parts of all other eigenvalues are strictly smaller.
It follows directly that the largest eigenvalue of $\A$ (denoted by $-\lambda_1$) is real and has $-\lambda_1>\Re (-\lambda_i)$
for all $\lambda_i\neq \lambda_1$.

However, as opposed to the single class case, the multiplicity of $-\lambda_1$ may be higher than 1. Also, there may be several
eigenvectors corresponding to $-\lambda_1$. This means that in order to calculate the contribution of $-\lambda_1$ to the pdf,
we need to be slightly more meticulous. The proof is essentially a transformation of the matrix $\A$ to a
form that is similar to the Jordan form (but not the same), while preserving some nonnegativity of $\A$ and $\av$ (where it is important).
We also present a numerical example at the end of this section to demonstrate the steps of the proof.

Let $\Q_i\J_i\Q_i^{-1}=\B_i$
be the Jordan decomposition of $\B_i$. We assume that the first block of $\J_i$ is the single dominant eigenvalue of $\B_i$;
Theorem \ref{thm:pf} thus guarantees that the first column of $\Q_i$, which is the corresponding right eigenvector,
is strictly positive, and the first row of $\Q_i^{-1}$, which is the corresponding left eigenvector, is also strictly positive.
Create the transformation matrix
$$\Q=\left[
\begin{array}{ccccc}
\Q_1 &  0 & 0 & \dots & 0\\
0 & \Q_2 & 0 & \dots & 0\\
\vdots & & & & \vdots \\
0 & & \dots & 0 & \Q_k
\end{array}\right].
$$
Then $\Q^{-1} \A \Q$ is an upper triangular matrix that contains the eigenvalues of $\A$ in its diagonal.
Applying this transformation to the pdf, we get
$$f_X(x)=-\av \A e^{\A x}\1=
-(\av \Q) (\Q^{-1}\A\Q)e^{(\Q^{-1}\A\Q)x} (\Q^{-1} \1).$$
Take all rows and columns of $\Q^{-1}\A\Q$ that have $-\lambda_1$ in the diagonal. Denote this submatrix by $\B$.
The submatrix $\B$ is responsible for the whole contribution of $-\lambda_1$. $\B$ can be calculated as
$$\B=\R \Q^{-1}\A\Q \R^T$$
where $\R$ is a $n_1\times n$ binary matrix (whose elements are either 0 or 1)
where $n_1$ is the multiplicity of the dominant eigenvalue in $\A$ and $n$
is the size of $\A$; row $i$ in $\R$ is equal to the unit vector $\ev_j$ if the $i$-th instance of $-\lambda_1$ in the
diagonal of $\Q^{-1}\A\Q$ is at coordinate $j,j$.
$(\av\Q)$ is strictly positive on the coordinates corresponding to $\B$ since the dominant eigenvector of $\Q_i$ are
strictly positive and the block of $\av$ associated with $\Q_i$ is nonnegative and different from 0
(if it was 0 then PH($\av$,$\A$) would be redundant). Similarly, $(\Q^{-1}\1)$ is strictly positive on
the coordinates corresponding to $\B$.

Finally, we argue that we can identify the dominant term in $e^{\B x}$ and see that it has a positive coefficient.
This is done directly instead of transforming $\B$ to Jordan form.
To this end, note that the offdiagonal elements of $\B$ are nonnegative since $\A$ originally contained nonnegative
elements above the diagonal, which were then multiplied by the strictly positive dominant left and right eigenvectors of each block $\B_i$.

The matrix $\lambda_1\I+\B$ is strictly upper triangular, thus nilpotent; this implies that the series expansion
$$e^{(\lambda_1\I+\B)x}=\sum_{k=0}^\infty \frac{((\lambda_1\I+\B)x)^k}{k!}$$
is actually a finite sum, and $e^{(\lambda_1\I+\B) x}$ is a polynomial of $x$. The dominant term in $e^{\B x}$ is
equal to the last nonzero term of this polynomial, multiplied by $e^{-\lambda_1x}$. The coefficient of this term
is necessarily positive since $(\lambda_1\I+\B)$ and thus powers of $(\lambda_1\I+\B)$ do not have negative elements.

Consequently, we have proved that $\lambda_1$ contributes to the pdf
$$f_X(x)=-\av \A e^{\A x}\1=
-(\av \A) (\Q^{-1}\A\Q)e^{(\Q^{-1}\A\Q)x} (\Q^{-1} \1).$$
with maximal multiplicity and with a positive coefficient, and the DEC holds.


\begin{example}
Let
$$\A=\left[\begin{array}{cccccccccc}
\cline{1-3}
\multicolumn{1}{|c}{-4} & 1 & \multicolumn{1}{c|}{1} & 0 & 0.2 & 0.4 & 0 & 0 & 0 & 0.4\\
\multicolumn{1}{|c}{1} & -2 & \multicolumn{1}{c|}{1} & 0 & 0 & 0 & 0 & 0 & 0 & 0\\
\multicolumn{1}{|c}{2} & 0 & \multicolumn{1}{c|}{-3} & 0 & 0 & 0 & 0 & 0.2 & 0.4 & 0.2\\
\cline{1-5}
0 & 0 & 0 & \multicolumn{1}{|c}{-4} & \multicolumn{1}{c|}{3} & 0.2 & 0.2 & 0 & 0.4 & 0\\
0 & 0 & 0 & \multicolumn{1}{|c}{1} & \multicolumn{1}{c|}{-2} & 0 & 0.2 & 0.2 & 0 & 0.2\\
\cline{4-7}
0 & 0 & 0 & 0 & 0 & \multicolumn{1}{|c}{-2} & \multicolumn{1}{c|}{1} & 0 & 1/5 & 0\\
0 & 0 & 0 & 0 & 0 &  \multicolumn{1}{|c}{1} & \multicolumn{1}{c|}{-2} & 0 & 0 & 0\\
\cline{6-9}
0 & 0 & 0 & 0 & 0 & 0 & 0 & \multicolumn{1}{|c}{-8} & \multicolumn{1}{c|}{2} & 0.6\\
0 & 0 & 0 & 0 & 0 & 0 & 0 & \multicolumn{1}{|c}{6} & \multicolumn{1}{c|}{-7} & 0 \\
\cline{8-10}
0 & 0 & 0 & 0 & 0 & 0 & 0 & 0 & 0 & \multicolumn{1}{|c|}{-1}\\
\cline{10-10}
\end{array}\right].$$
$\A$ has 5 communicating classes: $\B_1$ has size 3 and dominant eigenvalue $-1$, $\B_2$, $\B_3$ and $\B_4$ are of size 2
and their dominant eigenvalues are $-1, -1$ and $-4$ respectively; $\B_5$ is of size 1 with dominant eigenvalue $-1$. Thus $\lambda_1=1$.
$$\Q=\left[\begin{array}{cccccccccc}
\cline{1-3}
\multicolumn{1}{|c}1 & 0 & \multicolumn{1}{c|}{-1} & 0 & 0 & 0 & 0 & 0 & 0 & 0\\
\multicolumn{1}{|c}2 & -1 & \multicolumn{1}{c|}{0} & 0 & 0 & 0 & 0 & 0 & 0 & 0\\
\multicolumn{1}{|c}1 & 1 & \multicolumn{1}{c|}{1} & 0 & 0 & 0 & 0 & 0 & 0 & 0\\
\cline{1-5}
0 & 0 & 0 & \multicolumn{1}{|c}{1} & \multicolumn{1}{c|}{-3} & 0 & 0 & 0 & 0 & 0\\
0 & 0 & 0 & \multicolumn{1}{|c}{1} & \multicolumn{1}{c|}{1} & 0 & 0 & 0 & 0 & 0\\
\cline{4-7}
0 & 0 & 0 & 0 & 0 & \multicolumn{1}{|c}{1} & \multicolumn{1}{c|}{-1} & 0 & 0 & 0\\
0 & 0 & 0 & 0 & 0 & \multicolumn{1}{|c}{1} & \multicolumn{1}{c|}{1} & 0 & 0 & 0\\
\cline{6-9}
0 & 0 & 0 & 0 & 0 & 0 & 0 & \multicolumn{1}{|c}{1} & \multicolumn{1}{c|}{-2} & 0\\
0 & 0 & 0 & 0 & 0 & 0 & 0 & \multicolumn{1}{|c}{3} & \multicolumn{1}{c|}{2} & 0 \\
\cline{8-10}
0 & 0 & 0 & 0 & 0 & 0 & 0 & 0 & 0 & \multicolumn{1}{|c|}{1}\\
\cline{10-10}
\end{array}\right]$$
Notice that in $\Q$, the first column in each block is strictly positive. Even though it is not displayed in this example,
 $\Q$ (and $\Q^{-1}\A\Q$) may contain complex numbers, but only in rows and columns corresponding to non-dominant eigenvalues.
$$\Q^{-1}\A\Q=
\left[\begin{array}{cccccccccc}
\cline{1-1} \cline{4-4} \cline{6-6} \cline{10-10}
\multicolumn{1}{|c|}{-1} & 0 & 0 & \multicolumn{1}{|c|}{0.05} & 0.05 & \multicolumn{1}{|c|}{0.10} & -0.10 & 0.25 & 0.20 & \multicolumn{1}{|c|}{0.15}\\
\cline{1-1} \cline{4-4} \cline{6-6} \cline{10-10}
\multicolumn{1}{c}{0} & -3 & 0 & \multicolumn{1}{c}{0.10} & 0.10 & \multicolumn{1}{c}{0.20} & -0.20 & 0.50 & 0.40 & \multicolumn{1}{c}{0.30}\\
\multicolumn{1}{c}{0} & 0 & -5 & \multicolumn{1}{c}{-0.15} & 0.30 & \multicolumn{1}{c}{-0.15} & -0.30 & 0.25 & 0.20 & \multicolumn{1}{c}{-0.25}\\
\cline{1-1} \cline{4-4} \cline{6-6} \cline{10-10}
\multicolumn{1}{|c|}{0} & 0 & 0 & \multicolumn{1}{|c|}{-1} & 0 & \multicolumn{1}{|c|}{0.25} & 0.15 & 0.35 & 0 & \multicolumn{1}{|c|}{0.15}\\
\cline{1-1} \cline{4-4} \cline{6-6} \cline{10-10}
\multicolumn{1}{c}{0} & 0 & 0 & \multicolumn{1}{c}{0} & -5 & \multicolumn{1}{c}{-0.05} & 0.05 & -0.15 & -0.40 & \multicolumn{1}{c}{0.05}\\
\cline{1-1} \cline{4-4} \cline{6-6} \cline{10-10}
\multicolumn{1}{|c|}{0} & 0 & 0 & \multicolumn{1}{|c|}{0} & 0 & \multicolumn{1}{|c|}{-1} & 0 & 0.20 & 0.30 & \multicolumn{1}{|c|}{0}\\
\cline{1-1} \cline{4-4} \cline{6-6} \cline{10-10}
\multicolumn{1}{c}{0} & 0 & 0 & \multicolumn{1}{c}{0} & 0 & \multicolumn{1}{c}{0} & -3 & -0.20 & -0.30 & \multicolumn{1}{c}{0}\\
\multicolumn{1}{c}{0} & 0 & 0 & \multicolumn{1}{c}{0} & 0 & \multicolumn{1}{c}{0} & 0 & -4 & 0 & \multicolumn{1}{c}{9/35}\\
\multicolumn{1}{c}{0} & 0 & 0 & \multicolumn{1}{c}{0} & 0 & \multicolumn{1}{c}{0} & 0 & 0 & -11 & \multicolumn{1}{c}{-6/35} \\
\cline{1-1} \cline{4-4} \cline{6-6} \cline{10-10}
\multicolumn{1}{|c|}{0} & 0 & 0 & \multicolumn{1}{|c|}{0} & 0 & \multicolumn{1}{|c|}{0} & 0 & 0 & 0 & \multicolumn{1}{|c|}{-1}\\
\cline{1-1} \cline{4-4} \cline{6-6} \cline{10-10}
\end{array}\right].
$$
The rows and columns that include the dominant eigenvalue are marked and so
$$\R=\left[\begin{array}{cccccccccc}
1 & 0 & 0 & 0 & 0 & 0 & 0 & 0 & 0 & 0\\
0 & 0 & 0 & 1 & 0 & 0 & 0 & 0 & 0 & 0\\
0 & 0 & 0 & 0 & 0 & 1 & 0 & 0 & 0 & 0\\
0 & 0 & 0 & 0 & 0 & 0 & 0 & 0 & 0 & 1\\
\end{array}\right],
~~
\B=\R \Q^{-1}\A\Q \R^T=\left[\begin{array}{cccc}
-1 & 0.05 & 0.10 & 0.15\\
0 & -1 & 0.25 & 0.15\\
0 & 0 & -1 & 0\\
0 & 0 & 0 & -1
\end{array}\right].$$
The last nonzero power of the nilpotent matrix $\lambda_1 \I+\B$ is
$$(\lambda_1 \I+\B)^2=
\left[\begin{array}{cccc}
0 & 0 & 0.00125 & 0.0075\\
0 & 0 & 0 & 0\\
0 & 0 & 0 & 0\\
0 & 0 & 0 & 0
\end{array}\right]$$
whose nonzero elements are all positive.

\end{example}

\section{Proofs for the sufficient direction}
\label{sec:appc}

\subsection{Proof of Lemma \ref{lemma:fnulla}.}

The intuitive behavior of the convolution of
the pdf of a non-negative r.v. ($Y$) and the $\textrm{Erlang}(l,\mu)$ pdf
is the following: assume $f_Y(0)>0$; for large values of $\mu$, the
Erlang pdf decays rapidly, so the function $f_Y$ is very close to
$f_X$, except around $0$, since convolution of a pdf $f_Y$ with an $\textrm{Erlang}(l,\mu)$ pdf increases the multiplicity of 0 by $l$. Lemma \ref{lemma:fnulla} utilizes this relation in the opposite direction.
Hence if $f_X$ was positive everywhere except at 0 with multiplicity $l$, then $f_Y$ will be positive at 0 and its positivity for $\mathbb{R}^+$ comes from the small difference from $f_X$. { (Actually, the tail and the main body of $f_Y(x)$ will be examined separately for technical reasons.)}

$f_Y$ can be calculated in the Laplace-transform domain as follows.
The Laplace-transform of the $\textrm{Erlang}(l,\mu)$ pdf is
$$f^*_{k,\mu}(s)=\left(\frac{\mu}{s+\mu}\right)^l.$$

Denote by $f_X^* (s)$ and $f_Y^* (s)$ the Laplace-transform of $f_X$ and $f_Y$, respectively. Then
from $f_X(x)=f_Y(x)\ast f_{l,\mu}(x)$ we have $f_X^* (s)=f_Y^* (s)\cdot \left(\frac{\mu}{s+\mu}\right)^l$, and so
$$f_Y^* (s)=f_X^* (s)\left(\frac{s+\mu}{\mu}\right)^l=
f_X^* (s)\left(1+\frac{s}{\mu}\right)^l.$$

For $l=1$, the inverse transform of $f_X^* (s)\left(1+\frac{s}{\mu}\right)$ gives
$$f_Y(x)=f_X(x)+\frac{1}{\mu}\left(f_X'(x)+f_X(0)\right)=f_X(x)+\frac{1}{\mu}f_X'(x).$$

For $l>1$, induction (or the binomial theorem) gives
$$f_Y(x)=\sum_{i=0}^l \binom{l}{i}\frac{1}{\mu^i}f_X^{(i)}(x)=
-\gv\sum_{i=0}^l \binom{l}{i}\left(\frac{\G}{\mu}\right)^{i} \G e^{\G x}\ones.$$

The fact that $f_Y(x)$ is a matrix-exponential pdf is
straightforward from the above formula. { Also, it has a representation of the form ME$(\gv',\G)$ where $\gv'=-\gv\sum_{i=0}^l \binom{l}{i}\left(\frac{\G}{\mu}\right)^{i}$.}

{ We fix a value $\delta>0$ (independent from $\mu$) such that
$$f_X^{(l)}(x)>0, \quad x\in (0,\delta].$$

This is possible since $f_X^{(l)}(0)>0$ and $f_X^{(l)}$ is continuous. This in turn implies by integration that
$$f_X^{(i)}(x)\geq 0 , \quad x\in (0,\delta].$$
for every $i=l,l-1,\dots,1,0$ and thus
$$f_Y(x)=\sum_{i=0}^l \binom{l}{i}\frac{1}{\mu^i}f_X^{(i)}(x)>0\quad x\in(0,\delta).$$
This holds for any value of $\mu$.

We examine the tail of $f_Y$ next. Recall that as $x\to\infty$, $f_X(x)$ decays as $c_{\lambda_1,n_1}x^{n_1-1}e^{-\lambda_1x}$ where $c_{\lambda_1,n_1}>0$.

\begin{align}
\label{Ybinom}
f_Y(x)-f_X(x)=\sum_{i=0}^l
\binom{l}{i}\frac{1}{\mu^i}f_X^{(i)}(x)-f_X(x)= -\gv\sum_{i=1}^l
\binom{l}{i}\left(\frac{\G}{\mu}\right)^i \G{e^{\G x}}\ones.
\end{align}
Since $e^{\G x}$ decays with rate $x^{n_1-1}e^{-\lambda_1x}$,
$$-\gv \binom{l}{i}{\G}^i \G{e^{\G x}}\ones\sim c_i x^{n_1-1}e^{-\lambda_1x}$$
for each $i=1,\dots,l$ for some constants $c_i$.

Select $K_1$ such that
$$\left|\frac{-\gv \binom{l}{i}{\G}^i \G{e^{\G x}}\ones}{x^{n_1-1}e^{-\lambda_1x}}\right|\leq 2|c_i| \quad \forall x>K_1$$
for $i=1,\dots,k$, Then
$$|f_Y(x)-f_X(x)|\leq \sum_{i=1}^l \frac{2|c_i|}{\mu^i}{x^{n_1-1}e^{-\lambda_1x}}.$$
Note that $K_1$ is also independent from $\mu$.

The constant $\sum_{i=1}^l \frac{2|c_i|}{\mu^i}$ is decreasing in $\mu$. Select $\mu_0$ such that
$$\sum_{i=1}^l \frac{2|c_i|}{\mu^i}\leq \frac{1}{2}c_{\lambda_1,n_1}\quad
\forall \mu>\mu_0.$$

Select $K_2$ such that
$$f_X(x)\geq \frac{1}{2}c_{\lambda_1,n_1}{x^{n_1-1}e^{-\lambda_1x}}\quad \forall x>K_2.$$

Set $K=\max(K_1,K_2)$. At this point, $\delta$ and $K$ are fixed (independently of $\mu$), and for any $\mu>\mu_0$ it holds that
$$f_Y(x)\geq f_X(x)-|f_Y(x)-f_X(x)|\geq \frac{1}{2}c_{\lambda_1,n_1}{x^{n_1-1}e^{-\lambda_1x}}-\frac{1}{2}c_{\lambda_1,n_1}{x^{n_1-1}e^{-\lambda_1x}}=0 \quad\forall x>K.$$

We now have positivity of $f_Y$ at $[0,\delta]$ and $[K,\infty]$. For $[\delta,K]$, we use the formula \eqref{Ybinom} again, and note that

\begin{align*}&\sup_{x\in[\delta,K]}\left|\gv\sum_{i=1}^l
\binom{l}{i}\left(\frac{\G}{\mu}\right)^i \G{e^{\G x}}\ones\right|
\leq\sum_{i=1}^l\left(\frac{1}{\mu}\right)^i\sup_{x\in[\delta,K]}
\left|\gv\binom{l}{i}\G^i \G{e^{\G x}}\ones\right|,
\end{align*}
where $\sup_{x\in[\delta,K]}\left|\gv\binom{l}{i}\G^i \G{e^{\G x}}\ones\right|$ is finite for each $i=1,\dots,l$, while $\frac{1}{\mu^i}\to 0$, so there exists a $\mu_1$ such that for any $\mu>\mu_1$,
$$|f_Y(x)-f_X(x)|\leq \inf_{x\in[\delta,K]} f_X(x),$$
which is positive due to the positive density condition (specifically that $f_X$ is strictly positive on a finite interval not containing 0).

Selecting any $\mu>\max(\mu_1,\mu_2)$ finishes the lemma.

}

\subsection{Proof of Lemma \ref{lemma:kozelites}.}

We will prove part \emph{i)} first.

We will begin by showing that the supremum is obtained for $z=r$.

Series expansion gives
$$e^z-\left(1+\frac{z}{n}\right)^n=\sum_{k=0}^\infty \frac{z^k}{k!}B(n,k),$$
where
$$B(n,k)=\left\{
\begin{array}{cc}
1-\frac{n(n-1)\dots(n-k+1)}{n^k}\quad & \textrm{ if } k\leq n\\
1 & \textrm{ if } k > n
\end{array}
\right.
$$

Note the following properties of $B(n,k)$:
$$0\leq B(n,k)\leq 1 \,\,\forall\, n,\, k; \qquad \lim_{n\to\infty}B(n,k)=0 \,\,\forall\, k.$$

For every $z$ with $|z|\leq r$, we have
$$\left|e^z-\left(1+\frac{z}{n}\right)^n\right|=
\left|\sum_{k=0}^\infty \frac{z^k}{k!}B(n,k)\right|\leq
\sum_{k=0}^\infty \frac{|z|^k}{k!}B(n,k)\leq
\sum_{k=0}^\infty \frac{r^k}{k!}B(n,k)=
\left|e^r-\left(1+\frac{r}{n}\right)^n\right|.$$

Notice that the series expansion ensures $e^r-\left(1+\frac{r}n\right)^n>0$, so we only
need an upper bound on $e^r-\left(1+\frac{r}n\right)^n$. Using the straightforward inequalities
$$\ln(1+x)\geq x-\frac{x^2}{2}\,\, (x\geq 0)\quad \textrm{and}\quad
e^x\geq 1+x\,\, (x\in \mathbb{R})$$
we get that
$$e^r-\left(1+\frac{r}n\right)^n=e^r-e^{n\ln(1+r/n)}\leq
e^r-e^{r-r^2/(2n)}= e^r\left(1-e^{-r^2/(2n)}\right)\leq
e^r\left(1-\left(1-\frac{r^2}{2n}\right)\right)=e^r\
\frac{r^2}{2n}.$$

We note that this estimate is asymptotically sharp as $n\to\infty$.

For part \emph{ii)}, we use the series expansion again:
\begin{eqnarray*}
\left\|e^{\H}-\left(1+\frac{\H}{n}\right)^n\right\| &=&
\left\|\sum_{k=0}^\infty \frac{\H^k}{k!}B(n,k)\right\|\leq
\sum_{k=0}^\infty \frac{\|\H\|^k}{k!}B(n,k)\leq \\
&\leq&
\sum_{k=0}^\infty \frac{r^k}{k!}B(n,k)=e^r-\left(1+\frac{r}{n}\right)^n \leq \frac{r^2e^r}{2n},
\end{eqnarray*}
where $r= \| \H \|$.

\subsection{Proof of Lemma \ref{lemma:csillag}.}
\label{subsection:csillag}

According to the FE block composition of $\G$ it has the following block structure
\begin{equation}\label{eq:strG}
\G=\left[
 \begin{array}{c|c}\G_{11} & \G_{12}\\ \hline 0 & \G_{22} \end{array}\right]~,
\end{equation}
where
$$\G_{11}=\left[
\begin{array}{ccccc}
-\lambda_1 & \lambda_1 & 0 & \dots & 0\\
0 & -\lambda_1 & \lambda_1 & \dots & 0\\
\vdots\\
0 & & \dots & 0 & -\lambda_1
\end{array}\right], ~~~
\G_{12}=\left[\begin{array}{cccc}
0 & 0 & \dots & 0\\
\vdots & & & \vdots\\
0 & 0 & \dots & 0\\
\lambda_1 & 0 & \dots & 0
\end{array}\right]~,
$$
and $\G_{22}$ contains the rest of the FE blocks. The size of $\G_{11}$ is denoted by $n_1$
(which is the multiplicity of the dominant eigenvalue $-\lambda_1$)
and the size of $\G_{22}$ by $n_2$.
Let
$$\H=\G+\lambda_1\I,$$
and accordingly $\H_{11}=\G_{11}+\lambda_1\I, \H_{12}=\G_{12}$ and
$\H_{22}=\G_{22}+\lambda_1\I$, where $\I$ denotes the identity matrix
of appropriate size.
From $\H=\G+\lambda_1\I,$ it follows that
\[ e^{\G x} = e^{-\lambda_1 x} e^{\H x}, \]
and it is enough to investigate the dominant row of $e^{\H x}$.
In the rest of the proof, $(.)_{11},(.)_{12},(.)_{22}$
denote the corresponding matrix blocks (\emph{not} single elements).
The eigenvalues of
$\H_{22}$ have negative real parts. Their real parts are less than or equal
to $\lambda_1 - \Re(\lambda_2)$, where $-\lambda_2$ is the eigenvalue with the second largest real part.

From the series expansion of $e^{\H x}$
$$e^{\H x}=\sum_{n=0}^\infty \frac{x^n}{n!}\H^n,$$
and the block triangular structure of $\H$ we have that the upper left block is
$$(e^{\H x})_{11}=\sum_{n=0}^\infty \frac{x^n}{n!}\H_{11}^n,$$
where $(\H_{11}x)^n$ can be calculated explicitly:
$$(\H_{11}x)^n=
\left[\begin{array}{ccccccc}
0 & \dots & 0 & (\lambda_1 x)^n & 0 & \dots & 0\\
0 & \dots & 0 & 0 & (\lambda_1 x)^n & \dots & 0\\
\vdots \\
0 & & & \dots & & & (\lambda_1 x)^n\\
0 & & & & & & 0\\
\vdots\\
0 & & & \dots & & & 0\\
\end{array}\right]$$
with the nonzero elements being at positions
$(1,n+1),(2,n+2),\dots$. Specifically, $\H_{11}^n$ is $0$ for $n\geq
n_1$, so  the sum $\sum_{n=0}^\infty \frac{x^n}{n!}\H_{11}^n,$ is
actually finite, and from the above form it is clear that $(e^{\H
t})_{11}$ is upper diagonal, dominated by its first row, which of
course also dominates $(e^{\H x})_{21}=0$.

The rest of the proof is devoted to the elements of  $(e^{\H x})_{12}$ and $(e^{\H x})_{22}$.
For that, we need to examine $(e^{\H x})_{12}$.0
$$(e^{\H x})_{12}=\sum_{n=0}^\infty \frac{x^n}{n!} (\H^n)_{12}.$$
Here,
$$(\H^n)_{12}=\sum_{k=0}^{n-1}(\H_{11})^k\H_{12}(\H_{22})^{n-k-1}$$
since $\H$ is an upper block bi-diagonal matrix.
Thus
\begin{align*}
(e^{\H x})_{12}
&=\sum_{n=1}^\infty \frac{x^n}{n!}\sum_{k=0}^{n-1}(\H_{11})^k\H_{12}(\H_{22})^{n-k-1}\\
&=\sum_{k=0}^\infty(\H_{11})^k\H_{12}\sum_{n=k+1}^\infty
\frac{x^n}{n!}(\H_{22})^{n-k-1}\\
&=\sum_{k=0}^{n_1-1}(\H_{11})^k\H_{12}\sum_{n=k+1}^\infty
\frac{x^n}{n!}(\H_{22})^{n-k-1}.
\end{align*}
Again, the sum over $k$ is finite.

The inner sum can be calculated as
$$\sum_{n=k+1}^\infty \frac{1}{n!}x^{n-k-1}=x^{-k-1}\sum_{n=k+1}^\infty \frac{1}{n!}x^n=
x^{-k-1}\left(e^x-\sum_{l=0}^k \frac{x^l}{l!}\right),$$
and accordingly,
$$\sum_{n=k+1}^\infty
\frac{x^n}{n!}(\H_{22})^{n-k-1}=
(\H_{22})^{-(k+1)}
\left(e^{x\H_{22}}-\mx{I}-\H_{22}x-\dots-\frac{(\H_{22}x)^k}{k!}
\right).$$

Putting it all together, we obtain that
$$(e^{\H x})_{12}=
\sum_{k=0}^{n_1-1}(\H_{11})^k\H_{12}(\H_{22})^{-(k+1)}
\left(e^{x\H_{22}}-\mx{I}-\H_{22}x-\dots-\frac{(\H_{22}x)^k}{k!}
\right).$$

The form of $(\H_{11})^k\H_{12}$ guarantees that for each $k$
$$(\H_{11})^k\H_{12}(\H_{22})^{-(k+1)}
\left(e^{x\H_{22}}-\mx{I}-\H_{22}x-\dots-\frac{(\H_{22}x)^k}{k!}
\right)$$
has a single nonzero row, with $k=n_1-1$ corresponding to the first row being nonzero, $k=n_1-2$ to the second etc.\ Within each row, the main term is
$$-(\H_{11})^k\H_{12}(\H_{22})^{-(k+1)}
\frac{(\H_{22}x)^k}{k!}=-\frac{x^k}{k!}(\H_{11})^k\H_{12}(\H_{22})^{-1}.$$

Specifically, the main term in each element of the first row is  of
order $x^{n_1-1}$, and the order in the other rows is smaller within
the block $(e^{\H x})_{12}$.

We need to calculate $\H_{22}^{-1}$. It can be calculated either
via Cramer's (which allows for calculating the constants $C_j$
explicitly, but is left to the reader), or by using the following
identity:
$$\H_{22}^{-1}=
-\int_{\tau=0}^\infty e^{\H_{22}x}\mathrm{d} t=
-\int_{\tau=0}^\infty e^{\lambda_1 x}\cdot e^{\G_{22}x}\mathrm{d} t.$$

The integral exists because all eigenvalues of $\H_{22}$ have
negative real part. $e^{\lambda_1 x}$ is a positive function
(``weight'') and $e^{\G_{22}x}$ contains the transition
probabilities of a CTMC, so all elements of $e^{\G_{22}x}$ are
positive for all $t>0$. Thus all elements of $\H_{22}^{-1}$ are
negative, and the single nonzero row of
$-(\H_{11})^k\H_{12}(\H_{22})^{-(k+1)} \frac{(\H_{22}x)^k}{k!}$ is
strictly positive.

Finally, since the block $(\H)_{22}$ has eigenvalues with
negative real part, the elements of $(e^{\H x})_{22}$ decay
exponentially, so they are of course dominated by the first row of
$(e^{\H x})_{12}$.

One last remark:
\cite[page 771]{Mocanu99} discusses
the same statement
in a rather descriptive manner
using communicating classes.
\end{document}